\documentclass[journal]{IEEEtran}
\IEEEoverridecommandlockouts

\ifCLASSINFOpdf
  \usepackage[pdftex]{graphicx}
\else
\fi
\usepackage{cite}
\usepackage[cmex10]{amsmath}
\usepackage{array}
\usepackage{amsthm}
\usepackage{amsmath}
\usepackage{amsfonts}
\usepackage{balance}
\usepackage[usestackEOL]{stackengine}
\usepackage{amssymb}
\usepackage[utf8x]{inputenc}
\usepackage[font=small,labelfont=bf]{caption}
\newtheorem{definition}{Definition}
\newtheorem{theorem}{Theorem}

\newtheorem{lemma}{Lemma}

\newtheorem{remark}{Remark}

\newtheorem{assumption}{Assumption}[section]
\allowdisplaybreaks

\begin{document}
%
\title{Sparse Packetized Predictive Control of Disturbed Plants Over Channels with Data Loss}
%
%
%

\author{Mohsen~Barforooshan,
	Masaaki~Nagahara,~\IEEEmembership{Senior Member,~IEEE,}
	Jan~\O stergaard,~\IEEEmembership{Senior Member,~IEEE}
\thanks{

M. Barforooshan and J. \O stergaard are with the Department of
Electronic Systems, Aalborg University, DK-9220, Aalborg, Denmark (email: mob@es.aau.dk; jo@es.aau.dk).

M. Nagahara is with the Institute of Environmental Science and Technology, The University of Kitakyushu, Fukuoka, 808-0139, Japan
(email: nagahara@kitakyu-u.ac.jp). His work is partly supported by JSPS KAKENHI Grant Numbers JP20H02172, JP20K21008, and JP19H02301.
}
}

\maketitle

\begin{abstract} 
This paper investigates closed-loop stability of a linear discrete-time plant subject to bounded disturbances when controlled according to packetized predictive control (PPC) policies. In the considered feedback loop, the controller is connected to the actuator via a digital communication channel imposing bounded dropouts. Two PPC strategies are taken into account. In both cases, the control packets are generated by solving sparsity-promoting optimization problems. One is based upon an $\ell^2$-constrained $\ell^0$ optimization problem. Such problem is relaxed by an $\ell^1$-$\ell^2$ optimization problem in the other sparse PPC setting. We utilize effective solving methods for the latter optimization problems. Moreover, we show that in the presence of plant disturbances, the $\ell^2$-constrained $\ell^0$ sparse PPC and unconstrained $\ell^1$-$\ell^2$ sparse PPC guarantee practical stability for the system if certain conditions are met. More precisely, in each case, we can derive an upper bound on system state if the design parameters satisfy certain conditions. The bounds we derive are increasing with respect to the disturbance magnitude. We show via simulation that in both cases of proposed sparse PPC strategies, larger disturbances bring about performance degradation with no effect on system practical stability.         	
\end{abstract}

%
\IEEEpeerreviewmaketitle


\section{Introduction}
\label{intro}
Considering communication constraints in the analysis and design of feedback control loops gave rise to the paradigm of networked control systems (NCSs)\cite{zhang2013network}. In an NCS, the components of the control system exchange information over communication channels. Inevitable imperfections of communication channels (time delays, data loss, quantization, etc.) bring new limitations to system performance\cite{zhang2019networked}. Analyzing such limitations and designing control strategies compensating for them is a crucial aim in the area of NCSs.

Along the above lines, packetized predictive control (PPC) is proposed to attain robustness towards channel uncertainties such as packet dropouts and time delays \cite{tang2006compensation,quevedo2007packetized}. Basically, PPC is a model predictive control (MPC) strategy in which at each time instant, the controller transmits the entire sequence of calculated control commands as a data packet to the plant\cite{quevedo2011packetized}. In an early attempt, \cite{quevedo2012robust} derives the conditions of stochastic stability for a nonlinear plant controlled based on PPC over a channel with Markovian data loss. In \cite{peters2016shaped}, the authors investigate how fixed-rate vector quantization affect mean-square error (MSE) performance  in PPC. Recently, \cite{da2020hybrid} applies hybrid automatic repeat request (HARQ) technique to PPC for reducing the data transmission power.

There are several advantages associated with minimizing the control effort in feedback control applications. Such merits are mostly environmental (lowering pollution\cite{dunham1974automatic}), economical (energy saving\cite{xiao2020modeling}) and technical (reducing noise and vibration\cite{anderson2007optimal}). Along the lines of control effort minimization, \cite{nagahara2016maximum} proposes maximum hands-off control. In this approach, the goal is obtaining as many zero-valued control commands as possible over the time while maintaining a certain performance level for the closed-loop system\cite{nagahara2016discrete}. Hence, the maximum hands-off controller solves a sparsity-promoting optimization problem subject to performance constraints at each time instant\cite{chatterjee2016characterization}. Maximum hands-off control has been studied in various control settings over the literature. For instance, it is employed to attain consensus in multi-agent systems\cite{ikeda2018maximum}. The maximum hands-off control for linear plants with polytopic uncertainties is explored in \cite{kishida2018hands}. As another example, \cite{ikeda2020maximum} investigates maximum hands-off control problem under space-sparsity constraints.

The ideas of PPC and maximum hands-off control coexist in sparse PPC. More specifically, sparsity-promoting cost functions are utilized in sparse PPC to generate the control packets\cite{nagahara2011sparse,nagahara2012packetized}. Recall that in PPC, the control packets are calculated by minimizing finite-horizon cost functions. Sparse PPC for linear disturbance-free plants over delay-free communication channels with bounded dropouts is investigated in \cite{nagahara2014sparse}. The corresponding sparsity-promoting optimization problems in \cite{nagahara2014sparse} are ${\ell}^2$-constrained ${\ell}^0$ and a relaxed version of it; unconstrained ${\ell}^1$-${\ell}^2$ problems. For  ${\ell}^2$-constrained ${\ell}^0$ sparse PPC, \cite{nagahara2014sparse} derives the conditions of asymptotic stability. However, according to \cite{nagahara2014sparse}, only practical stability can be guaranteed for the unconstrained ${\ell}^1$-${\ell}^2$ sparse PPC. Furthermore, effective algorithms are proposed in \cite{nagahara2014sparse} for optimizing the considered sparsity promoting cost functions. In our recent contribution \cite{barforooshan2019sparse}, we extended the results of  \cite{nagahara2014sparse} to the case with constant channel delays.

In this paper, we analyze the stability of an NCS wherein a fully observable discrete-time linear plant exchanges data with a controller operating based on sparse PPC. The plant is subject to bounded disturbances. No communication constraint is considered for the path between the sensor and the controller. However, data packets produced by the controller might drop in the channel before reaching the actuator. The number of consecutive packet dropouts is assumed to be bounded. We consider a PPC strategy based on ${\ell}^2$-constrained ${\ell}^0$ sparsity-promoting optimization. As another control policy, we consider unconstrained ${\ell}^1$-${\ell}^2$ sparse PPC which performs based on a relaxation of ${\ell}^2$-constrained ${\ell}^0$ optimization. We show in each case that under certain conditions, practical stability can be achieved for the considered NCS. Motivated by their strength, we inspire from the approaches employed in \cite{nagahara2014sparse} when anlalyzing the stability and finding algorithms to solve the involved optimization problems. Nevertheless, as opposed to \cite{nagahara2014sparse}, the plant is disturbed here. The $\ell^2$ norm of the considered disturbance signal is bounded from above by a finite value. So, the first contribution of our work is that we propose sufficient conditions under which the ${\ell}^2$-constrained ${\ell}^0$ sparse PPC and unconstrained ${\ell}^1$-${\ell}^2$ sparse PPC render the system practically stable, despite the presence of plant disturbances and channel data loss. This leads to proposing a control design approach which is the second contribution of this paper. The third contribution is to shed lights on the performance loss introduced by the existence of plant disturbance. This is nailed by showing that the upper bounds derived on the $\ell^2$ norm of the plant state are increasing with respect to the disturbance upper bounds in both cases of ${\ell}^2$-constrained ${\ell}^0$ sparse PPC and unconstrained ${\ell}^1$-${\ell}^2$ sparse PPC. We illustrate our findings via a simulation example. In this example, we demonstrate that the practical stability is obtained with sparse control inputs by the proposed ${\ell}^2$-constrained ${\ell}^0$ sparse PPC and unconstrained ${\ell}^1$-${\ell}^2$ sparse PPC. Moreover, we observe from the simulation example that while stability is intact, system performance degrades when disturbance magnitude grows.

The paper is outlined as follows. The notation is provided by Section~II. The general PPC problem formulation is given in Section~III. Section~IV and Section~V formalize the ${\ell}^2$-constrained ${\ell}^0$ sparse PPC and unconstrained ${\ell}^1$-${\ell}^2$ sparse PPC with their corresponding stability analyzes, respectively. The simulation example is presented in Section~VI. Section X draws the conclusions of the paper.   

\section{Notation}
The set of natural numbers is symbolized by $\mathbb{N}$ based on which ${\mathbb{N}_{0}}$ is defined as ${\mathbb{N}_{0}}\triangleq\mathbb{N}\cup\{0\}$. We denote the transpose of matrix (vector) $M$ by $M^{\top}$. Moreover, $I_{n\times{n}}$ represents an identity matrix with dimensions $n\times{n}$ where $n\in{\mathbb{N}}$. For the vector $z=[{z_1},\dots,{z_n}]^{\top}\in{{\mathbb{R}}^{n}}$ in the euclidean space, $\ell^1$, $\ell^2$, and $\ell^\infty$ norms are defined as ${{\lVert{z}\rVert}_{1}}\triangleq{{|z_1|}+\dots+{|z_n|}}$, ${{\lVert{z}\rVert}_{2}}\triangleq{\sqrt{{z^\top}{z}}}$ and ${{\lVert{z}\rVert}_{\infty}}\triangleq{\max\{{|z_1|},\dots,{|z_n|}\}}$, respectively. The $\ell^0$ norm of the vector $z$ is the number of its non-zero elements. Furthermore, ${{\lVert{z}\rVert}_{W}}$ presents the weighted norm of the vector $z$ with respect to the positive definite matrix $W>0$ and is defined as ${{\lVert{z}\rVert}_{W}}\triangleq{\sqrt{{z^\top}W{z}}}$. We symbolize the minimum and maximum eigenvalues of the Hermitian matrix $W$ by  ${{{\lambda}_{\min}}(W)}$ and ${{{\lambda}_{\max}}(W)}$, respectively. Moreover, ${{\sigma}_{\max}}(W)$ is defined as ${{\sigma}_{\max}}(W)\triangleq\sqrt{{{{\lambda}_{\max}}({W^\top}W)}}$.
\begin{figure}
	\centering
	\includegraphics[width=7cm,height=5cm]{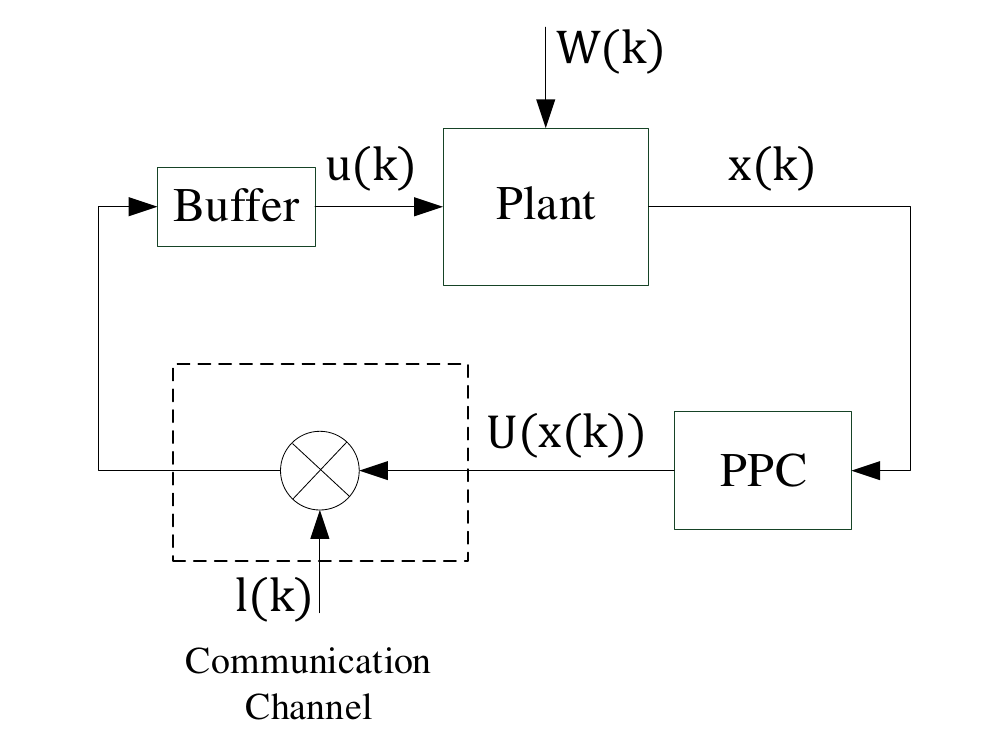}
	\caption{Considered PPC system}
	\label{ppcfig1}
\end{figure} 
\section{PPC for disturbed linear plants}\label{PDL}
The control setting of our interest is depicted by Fig.~\ref{ppcfig1}. One component of the NCS of Fig.~\ref{ppcfig1} is a linear discrete-time plant which is modeled as
\begin{equation}\label{eq1}
x(k+1)=Ax(k)+Bu(k)+w(k),\qquad k\in{{\mathbb{N}}_0}.
\end{equation}  
 In (\ref{eq1}), $x(k)\in{{\mathbb{R}}^n}$ and $u(k)\in{{\mathbb{R}}}$ are the plant state and the control input, respectively. Moreover, $w(k)\in{{\mathbb{W}}}\subset{{\mathbb{R}}^n}$ symbolizes an uncertain disturbance signal where ${{\mathbb{W}}}$ is a compact set for which $0\in{\mathbb{W}}$ holds. It stems from the properties of ${{\mathbb{W}}}$ that 
 \begin{equation}\label{eq1nn}
 {{\lVert{w(k)}\rVert}_{2}}\leq{W_m},\qquad k\in{{\mathbb{N}}_0},
 \end{equation}
 where ${W_m}\triangleq{\max_{\eta\in{\mathbb{W}}}{{{\lVert{\eta}\rVert}_{2}}}}$. So, the $\ell^2$-norm of the disturbance signal is bounded from above at each time instant by a finite value. Furthermore, $A$ and $B$ are state and input matrices. Such matrices are of appropriate dimensions, time-invariant and reachable.

As seen in the feedback loop of Fig.~\ref{ppcfig1}, the plant communicates with the controller over the actuation path via a digital communication channel. This channel is assumed to impose dropouts on the transmitted data. We model such data packet dropouts by the binary random sequence $l(k)$, that is, if $l(k)=1$, the data packet transmitted at time $k$ is received by the actuator at the same time. Otherwise, $l(k)=0$ and the packet is lost. The value of $l(k)$ is supposed to be unknown to the packetized predictive controller in Fig.~\ref{ppcfig1} at any time instant $k\in\mathbb{N}_0$. The packetized predictive controller operates via generating a sequence of control inputs at each time $k\in{\mathbb{N}_0}$. Let us describe such control input vector as follows:  
\begin{equation}\label{eq1n}
{U}({x}(k))=[{u_0}({{x}}(k))\dots {{u}_{N-1}}({{x}}(k))]^\top.
\end{equation}
The controller calculates ${U}({x}(k))$ by using the plant state $x(k)$ . Afterwards, it sends ${U}({x}(k))$ as a data packet to the plant over the communication channel. Suppose that ${U}({x}(k))$ arrives at the plant side of the channel at time $k\in{\mathbb{N}_0}$. First, ${U}({x}(k))$ is stored in a buffer next to the actuator in a way that the previous contents of the buffer are completely overwritten. Then, the actuator applies ${u_0}({{x}}(k))$ to the plant as the control input. If the packet ${U}({x}(k+1))$ is received one time instant later, then ${U}({x}(k+1))$ replaces $U(x(k))$ in the buffer and the actuator applies ${u_0}({{x}}(k+1))$ to the plant. Otherwise, ${u_1}({x}(k))$ will be the control input. Selecting the remaining elements of ${U}({x}(k))$ (${u}_n(x(k))$, $n\geq{2}$) in a successive manner and applying them as the control inputs continues until the arrival of a new control packet (${U}(x(k+n))$, $n\geq{2}$).    

\begin{assumption}\label{ass1}
The maximum number of consecutive packet dropouts is equal to $N-1$. It implies that there is always a fresh control input in the buffer to be applied to the plant. Moreover, the first control packet, $U(x(0))$, is received at the plant side.   
\end{assumption}
As already mentioned, in PPC, the control packets are solutions to the problem of minimizing a finite-horizon cost function at each time step. Suppose that the time is $k\in{\mathbb{N}_0}$ and ${x}={{x}}(k)$ is received by the controller. The PPC cost function has the following formulation:
\begin{equation}\label{eq2}
J({x},{U})=T({\tilde{x}_{N}})+\sum_{i=0}^{N-1}{S({\tilde{x}_{i}},{u_{i}})},
\end{equation}   
in which ${\tilde{x}_{i}}$ denotes a prediction of $x(k+i)$ with the horizon $N$ and $\{u_{i}\}_{i=0}^{N-1}$ are tentative future control inputs. The state prediction is carried out based on the update rule ${\tilde{x}_{i+1}}=A{\tilde{x}_{i}}+B{u_{i}}$ for every $i\in\{0\dots{N-1}\}$ where $\tilde{x}_{0}={{x}}(k)$. So, we use the disturbance-free (nominal) model of the plant for predicting the future states. This is due to the fact that the actual value of $w(k)$ is unknown at every time instant $k\in{\mathbb{N}_0}$. The functions $T$ and $S$ in (\ref{eq2}) are said to be terminal cost and stage cost, respectively.

We seek sparse solutions to the problem of optimizing $J({x},{U})$ in (\ref{eq2}). Towards this goal, we set $J$ together with a constraint in a way that an ${\ell}^2$-constrained ${\ell}^0$ sparsity-promoting optimization gives the control packets. Furthermore, we analyze system stability under the latter sparse PPC policy. The following section provides the details.  
\section{${\ell}^2$-Constrained ${\ell}^0$ Sparse PPC}
The cost function $J({x},{U})$ for ${\ell}^2$-constrained ${\ell}^0$
sparse PPC is specified by the terminal cost $T({\tilde{x}_{N}})=0$ and the stage cost $S$ which satisfies $S({\tilde{x}_{i}},{u_{i}})=0$ if ${u_{i}}=0$ and $S({\tilde{x}_{i}},{u_{i}})=1$ if $u_{i}\neq{0}$. The argument $U$ of the latter cost function is constrained to remaining in the following set:
 \begin{equation}\label{eq6}
 \upsilon(x)=\{{U}\in{}{{\mathbb{R}}^N}:{{\lVert{{\tilde{x}_{N}}}\rVert}_{P}^{2}}+\sum_{i=0}^{N-1}{{\lVert{{\tilde{x}_{i}}}\rVert}_{Q}^{2}}\leq{{{\lVert{{{x}}}\rVert}_{\Pi}^{2}}}\},
 \end{equation}  
where $x={\tilde{x}_{0}}$ and $U=[{u_{0}}\dots {{u}_{(N-1)}}]^\top$, as presented in the previous section. Moreover, matrices $P$, $Q$, and $\Pi$ are assumed to be positive definite and certify that for all ${x}\in{{\mathbb{R}}^n}$, $\upsilon(x)$ is non-empty. The latter pair of stage and terminal costs together with the set $\upsilon(x)$ in (\ref{eq6}) characterize an optimization problem that outputs the control packets in  ${\ell}^2$-constrained ${\ell}^0$ sparse PPC as follows:  
  \begin{align}\label{eq7}
 \begin{split}
 {U}(x)&=\arg\min_{U\in{\upsilon}(x)}{{{\lVert{{{U}}}\rVert}_{0}}}\\
 \upsilon(x)&=\{U\in{}{{\mathbb{R}}^N}:{{\lVert{M{U}-K{x}}\rVert}_{2}^{2}}\leq{{{\lVert{{{x}}}\rVert}_{\Pi}^{2}}}\}. 
 \end{split}
 \end{align}
 In (\ref{eq7}), $\upsilon(x)$ is reformulated in terms of the $\ell^2$ norm of a function of $x$ and $U$. Such restatement stems from the recursion ${\tilde{x}_{i+1}}=A{\tilde{x}_{i}}+B{u_{i}}$ used for predicting the future states $\{x(k+1),...,x(k+N)\}$ by $\{{\tilde{x}_{1}},...,{\tilde{x}_{N}}\}$ at any time instant $k\in{\mathbb{N}_0}$. Moreover, the matrices $M$ and $K$ are defined as $M\triangleq{\hat{Q}}^{\frac{1}{2}}\Gamma$ and $K\triangleq-{\hat{Q}}^{\frac{1}{2}}\Lambda$, respectively, where $\hat{Q}=\mathrm{diag}\{Q,\ldots,Q,P\}$ and  
\begin{equation}\label{eq4}
{\Gamma}=
\begin{bmatrix}
B & 0 & \ldots & 0\\
AB & B & \ldots & 0\\
\vdots & \vdots & \ddots & \vdots\\
{{A}^{N-1}}B & {{A}^{N-2}}B & \ldots & B 
\end{bmatrix},\qquad
\Lambda=
\begin{bmatrix}
A\\
{A^2}\\
\vdots\\
{A^N}
\end{bmatrix}.
\end{equation}
Due to its combinatorial properties and to avoid exhaustive search over a rather large discrete feasible set, we utilize orthogonal matching pursuit (OMP) algorithm to solve the optimization problem related to ${\ell}^2$-constrained ${\ell}^0$ sparse PPC \cite{nagahara2014sparse}. For such NP hard problem, the greedy iterative algorithm of OMP proves to be an effective method\cite{pati1993orthogonal}.     
\begin{remark}\label{rem1n}
The formulation of the $\ell^0$ optimization problem corresponding to (\ref{eq7}) in terms of $x$, $U$, and weighting matrices is based on the disturbance-free (nominal) model of the plant. So, the structure of the sparsity-promoting optimization considered here for ${\ell}^2$-constrained ${\ell}^0$ sparse PPC is the same as formulation derived in \cite[Section~III-B]{nagahara2014sparse} for ${\ell}^2$-constrained ${\ell}^0$ sparse PPC of disturbance-free LTI discrete-time plants.
\end{remark}
\noindent 
We will make use of the equivalence mentioned in Remark~\ref{rem1n} for deriving the conditions of practical stability in ${\ell}^2$-constrained ${\ell}^0$ sparse PPC of the NCS depicted by Fig.~\ref{ppcfig1}. 
\begin{definition}
	The NCS of Fig.~\ref{ppcfig1} is called practically stable if there exists $\varrho\in{{\mathbb{R}}^{+}}$ that satisfies $\lim\limits_{k\to\infty}{{\lVert{x(k)}\rVert}_{2}}\leq{\varrho}$. 
\end{definition} 
\noindent
We define the $i$-th iterated (open-loop) mapping ${f^i}(\cdot)$, $1\leq{i}\leq{N}$ with optimal vector ${U}$ as
\begin{equation}\label{eq11n2}
{f^i}(x,\bar{w}_0^{i-1})\triangleq{\bar{f}^i}(x)+{g^i}(\bar{w}_0^{i-1}),	
\end{equation}
where
\begin{align}\label{eq11n1}
\begin{split}
&{\bar{f}^i}(x)= {A^i}{{x}}+\sum_{l=0}^{i-1}{A^{i-1-l}}B{u_l}({x})\\
&{g^i}(\bar{w}_0^{i-1})=\sum_{l=0}^{i-1}{A^{i-1-l}}{\bar{w}_l}
\end{split}
\end{align}
 and $\bar{w}_{0}^{i-1}\triangleq[{\bar{w}_0}\dots{\bar{w}_{i-1}}]^\top$. We assume that $\bar{w}_0^{i-1}\in{\mathbb{W}^i}$, $\forall{i\in\{1,\dots,N\}}$ and ${f^0}(\cdot)={\bar{f}^0}(\cdot)=x$.  
 
\begin{lemma}\label{lemma0}
	For every realization of $\bar{w}_0^{i-1}\in{\mathbb{W}^i}$ and for every $i\in\{1,\dots,N\}$, the following holds:
	\begin{equation}\label{eq13n4n2}
	{\lVert{{g^i}(\bar{w}_0^{i-1})}\rVert}_{2}\leq{\gamma_N({{W}_m})},
	\end{equation}
	where
	\begin{equation}\label{eq13n4n2n}
	   {\gamma_N({{W}_m})}\triangleq{\sum_{l=0}^{N-1}{{{\sigma_{\max}(A^{N-1-l})}{{W}_m}}}}.	
	\end{equation}
\end{lemma}
\begin{proof}
It follows from (\ref{eq11n1}) and the triangle inequality that
\begin{equation}\label{eq13n4n}
{\lVert{{g^i}(\bar{w}_0^{i-1})}\rVert}_{2}\leq{\sum_{l=0}^{i-1}{{\lVert{A^{i-1-l}}{\bar{w}_l}}\rVert}_{2}}
\end{equation}
holds regardless of the realization of $\bar{w}_0^{i-1}\in{\mathbb{W}^i}$, $\forall{i\in\{1,\dots,N\}}$. Then, the claim follows from (\ref{eq1nn}), properties of weighted norms and the fact that the right-hand-side of (\ref{eq13n4n}) is an increasing function with respect to $i$.
\end{proof}
\begin{lemma}\label{lemma4}
Define ${{\upsilon}^{*}}(x)$ as
\begin{equation}\label{eq39}
{\upsilon^{*}}(x)\triangleq\{U\in{{\mathbb{R}}^N}:{{\lVert{MU-Kx}\rVert}_{2}^{2}}\leq{{{\lVert{{{x}}}\rVert}_{{\Pi}^*}^{2}}}\},
\end{equation}
where
\begin{equation}\label{eq10}
{{\Pi}^{\star}}\triangleq{{K^\top}(I-M{M^{\dagger}})M},\quad{M^\dagger}\triangleq{({M^\top}M)^{-1}{M^\top}}.
\end{equation}
Then ${\upsilon}(x)\supseteq{\upsilon^{*}}(x)$ holds if $\Pi\geq{\Pi^{*}}$, where $\Pi$ is related to ${\upsilon}(x)$ through (\ref{eq7}). Furthermore, given $\Pi\geq{\Pi^{*}}$, the feasible set	${\upsilon}(x)$ will be closed, convex and non-empty over ${{\mathbb{R}}^N}$.
\end{lemma}
\begin{proof}
One implication from Remark~\ref{rem1n} is that the set $\upsilon(x)$ and $\mathcal{U}(x)$ in \cite[Section~III-B]{nagahara2014sparse} have the same definitions as functions of $x$. Moreover, the matrix $\Pi^{*}$  has a structure identical to $W^*$ in \cite[Lemma 10]{nagahara2014sparse} considered for the disturbance-free system. Hence, we can conclude the claim immediately from \cite[Lemma~10]{nagahara2014sparse}.   
\end{proof}
\begin{lemma}\label{lemma5}
	Define the matrix $\xi>0$ as $\xi\triangleq{\Pi-{\Pi^{*}}}$. Then, there exists $\psi(x)\in{{\mathbb{R}}^N}$ for any feasible control packet $U\in{\upsilon}(x)$ in such a way that
		\begin{equation}\label{eq41}
	U={U^{*}}+\psi(x)
	\end{equation}
	holds where
	\begin{equation}\label{eq41s}
	{{\lVert{M\psi(x)}\rVert}_{2}^{2}}\leq{{\lVert{{{x}}}\rVert}_{\xi}^{2}}.
	\end{equation}
	and ${U^{*}}\in{{\upsilon^{*}}(x)}$. 	
\end{lemma}
\begin{proof}
	As noted in Remark~\ref{rem1n}, elements of the optimization problem (the cost function and constraint) are defined in the same way across (\ref{eq7}) and ${\ell}^2$-constrained ${\ell}^0$ sparse PPC of disturbance-free plants in \cite[Section ~III-B]{nagahara2014sparse}. Therefore, the claim results immediately from \cite[Lemma 11]{nagahara2014sparse}.
\end{proof}
\noindent
Let $P>0$ be the solution to the following Riccati equation:
\begin{equation}\label{eq12}
P={A^\top}PA-{A^\top}PB{{({B^\top}PB+r)}^{-1}}{B^\top}PA+Q,
\end{equation}
where $Q>0$ is chosen arbitrarily and $r=0$. Then according to \cite[Chapter 3]{bertsekas1976dynamic}, the elements of every feasible control packet ${U\in{\upsilon}(x)}$ are obtained via
\begin{equation}\label{eq41n}
{u_i}(x)={F(A+BF)^i}x+{\psi_i}(x),\quad i=0,1,\dots,N-1,
\end{equation} 
where ${\psi_i}(x)$ denotes the $i+1$-th element of ${\psi}(x)$ defined in Lemma~\ref{lemma5}. Moreover, the gain $F$ is calculated as
\begin{equation}\label{eq41n1}
F=-{({B^\top}PB)^{-1}{B^\top}PA}.
\end{equation}
In this case,  the $i$-th iterated (open-loop) mapping ${f^i}$ follows
\begin{align}\label{eq41n2}
\begin{split}
{f^{i+1}}(x,\bar{w}_0^{i})=(A+BF){f^{i}}(x,\bar{w}_0^{i-1})+B{\varsigma_i}(x)+{\varrho_i}(\bar{w}_0^{i}),
\end{split}
\end{align}	
where $\bar{w}_0^{i}\in{\mathbb{W}^{i+1}}$, $\forall{i\in\{0,\dots,N-1\}}$. Moreover, ${f^0}(\cdot)=x$ and the functions ${\varsigma_i}(x)$ and ${\varrho_i}(\bar{w}_0^{i})$ are defined as 
\begin{equation}\label{eq41n4n}
{\varsigma_i}(x)\triangleq{({B^\top}PB)^{-1}{B^\top}P}{\Gamma_i}{{\psi}(x)}
\end{equation}
and
\begin{equation}\label{eq41n4}
{\varrho_i}(\bar{w}_0^{i})\triangleq{\bar{w}_i}-\sum_{l=0}^{i-1}{BFA^{i-1-l}}{\bar{w}_l}
\end{equation}
for any $i\in\{0,\dots,N-1\}$, respectively. In (\ref{eq41n4n}), ${\Gamma_i}\in\mathbb{R}^{n\times{N}}$ denotes the $i+1$-th row block of matrix $\Gamma$ defined in (\ref{eq4}).
\begin{lemma}\label{lemma5n}
	For every realization of $\bar{w}_0^{i}\in{\mathbb{W}^{i+1}}$ and for every $i\in\{0,\dots,N-1\}$, the following holds:
	\begin{equation}\label{eq41n5}
	{\lVert{{\varrho_i}(\bar{w}_0^{i})}\rVert}_{2}\leq{\varepsilon_N({{W}_m})},
	\end{equation}
	where
\begin{equation}\label{eq41n5nn}
{\varepsilon_N({{W}_m})}\triangleq{\big[1+\sum_{l=0}^{N-1}{{{\sigma_{\max}(B{({B^\top}PB)^{-1}{B^\top}P}A^{N-l})}\big]{{W}_m}}}}.
\end{equation}
\end{lemma}
\begin{proof}
	It stems from (\ref{eq41n1}), (\ref{eq41n4}) and triangle inequality that 
	\begin{equation}\label{eqnew}
	{\lVert{{\varrho_i}(\bar{w}_0^{i})}\rVert}_{2}\leq{\lVert{{\bar{w}_i}}\rVert_{2}}+{{\sum_{l=0}^{i-1}\lVert{(B{({B^\top}PB)^{-1}{B^\top}P}A^{i-l})}{\bar{w}_l}}\rVert_{2}},
	\end{equation}
	where $i=0,1,\dots,N-1$ and $\bar{w}_0^{i}\in{\mathbb{W}^{i+1}}$. Now, the claim follows from (\ref{eq1nn}), properties of euclidean norms and the fact that the right-hand-side of (\ref{eqnew}) is an increasing function with respect to $i$.
\end{proof}  
\begin{lemma}\label{lemma6}
	Define ${V_P}(x)\triangleq{{\lVert{{{x}}}\rVert}_{P}^{2}}$. Choose $\Pi>0$ in such a way that $\Pi>{\Pi}^{*}$ holds. Set an arbitrary $Q>0$ and let $P>0$ be the solution of the Riccati equation (\ref{eq12}) with $r=0$. Then one can find constants $0\leq\varphi_1<1$ and $c_1>0$ certifying that
	\begin{equation}\label{eq42}
	\begin{split}
	\sqrt{{V_P}({f^{i+1}}(x,\bar{w}_0^{i}))}\leq&\Big[{{\varphi_1}}+\frac{{\sqrt{{c_1}\lambda_{\max}(\xi)}}}{(1-\varphi_1)\sqrt{\lambda_{\min}(P)}}\Big]\sqrt{{V_P}(x)}\\
	&+\big[(1-\varphi_1)^{-1}{\sqrt{\lambda_{\max}(P)}}\big]{\varepsilon_N({{W}_m})},
	\end{split}
	\end{equation}
 holds for every $x\in{\mathbb{R}^n}$ and every realization  $\bar{w}_0^{i}\in{\mathbb{W}^{i+1}}$, $i=1,\dots,N-1$. In (\ref{eq42}), $\xi$ is defined as in Lemma~\ref{lemma5}.   
\end{lemma}
\begin{proof}
	Let us define ${\Omega_i}(x,\bar{w}_0^{i-1})$ as
	\begin{align}\label{42n}
{\Omega_i}(x,\bar{w}_0^{i-1})\triangleq{(A+BF){f^{i}}(x,\bar{w}_0^{i-1})+B{\varsigma_i}(x)}
	\end{align} 	
 Then based on the triangle inequality of weighted norms \cite[Part I-Chapter~5]{poznyak2009advanced} and (\ref{eq41n2}), we have
	\begin{equation}\label{42n1}
	{{\lVert{{{f^{i+1}}(x,\bar{w}_0^{i})}}}\rVert}_{P}\leq{{\lVert{{{{\Omega_i}(x,\bar{w}_0^{i-1})}}}\rVert}_{P}+{\lVert{{{{\varrho_i}(\bar{w}_0^{i})}}}\rVert}_{P}}.	
	\end{equation}
	Considering Remark~\ref{rem1n}, we can conclude from the proof of \cite[Lemma~13]{nagahara2014sparse} and (\ref{eq41n5}) that 
	\begin{equation}\label{eq42n2}
	\begin{split}
		\sqrt{{V_P}({f^{i+1}}(x,\bar{w}_0^{i}))}\leq&\sqrt{\rho{V_P}({f^{i}}(x,\bar{w}_0^{i-1}))+{c_1}{{\lVert{{{x}}}\rVert}_{\xi}^{2}}}\\
		&+\Big({\sqrt{\lambda_{\max}(P)}}\Big){\varepsilon_N({{W}_m})}
	\end{split}
	\end{equation}
	is valid for every realization of $\bar{w}_0^{i}\in{\mathbb{W}^{i+1}}$ ,$\forall i\in\{0,1,\dots,N-1\}$. In (\ref{eq42n2}), $\rho$ and $c_1$ are defined as follows:
	\begin{equation}\label{eq42n3}
	\begin{split}
	\rho&\triangleq{1-\lambda_{\min}(QP^{-1})}\\
{c_1}&\triangleq{{\max_{i=0,1,\dots,N-1}}{\lambda_{\max}}}\Big({\Gamma_i^\top}P{\Gamma_i}(M^\top M)^{-1}\Big).
	\end{split}
	\end{equation}
	 Since $P\geq{Q}>0$, then $\rho\in[0,1)$ and ${c_1}>0$. Define $\varphi_1$ as $\varphi_1\triangleq\sqrt{\rho}$ and note that $f^0(\cdot)=x$. Now, based on the fact that $\sqrt{a+b}\leq{\sqrt{a}+\sqrt{b}}$, $\forall{a,b}\geq0$, and by mathematical induction, we can restate (\ref{eq42n2}) as 
	\begin{equation}\label{eq42n6}
	\begin{split}
		\sqrt{{V_P}({f^{i+1}}(x,\bar{w}_0^{i}))}\leq&{\varphi_1^i}\sqrt{{V_P}(x)}+(1+\dots+\varphi_{1}^{i-1})\\
	&\times\big[\sqrt{{c_1}}{{\lVert{{{x}}}\rVert}_{\xi}}+\Big({\sqrt{\lambda_{\max}(P)}}\Big){\varepsilon_N({{W}_m})}\big].
	\end{split}
	\end{equation}
Employing $0\leq\varphi_1<1$, properties of weighted norms and (\ref{eq42n6}), we derive (\ref{eq42}). The proof is complete now.
\end{proof}
\noindent
In the following theorem, the upper bound derived in Lemma~\ref{lemma6} will assist us establishing conditions of practical stability for the NCS of Fig.~\ref{ppcfig1} under $\ell^2$-constrained ${\ell}^0$ sparse PPC strategy presented by (\ref{eq7}) .   
\begin{theorem}\label{th2}
	 Select $Q>0$ arbitrarily and obtain $P>0$ as the solution to the Riccati equation (\ref{eq12}) with $r=0$. Choose the matrix $\xi$ in such a way that 
	\begin{equation}\label{eq42n8}
	\sqrt{\lambda_{\max}(\xi)}<\frac{(1-\varphi_1)^2\sqrt{\lambda_{\min}(P)}}{\sqrt{c_1}}
	\end{equation}  	
	holds where constants $\varphi_1\in[0,1)$ and $c_1>0$ are defined in the proof of Lemma~\ref{lemma6}. Calculate ${\Pi^*}$ as ${\Pi^*}=P-Q$ and set $\Pi$ as $\Pi={\Pi^*}+\xi$. Then, the  ${\ell}^2$-constrained ${\ell}^0$ sparse PPC characterized by (\ref{eq7}) gives control packets ${U}$ in such a way that ${{\lVert{x(k)}\rVert}_{2}}$ is bounded at each time instant $k\in{{\mathbb{N}}_0}$ and 
	\begin{equation}\label{eq42n8n}
	\lim\limits_{k\to\infty}{{\lVert{x(k)}\rVert}_{2}}\leq{\Psi_1},
	\end{equation}
	where
	\begin{equation}\label{eq42n8n1}
	\begin{split}
	{\Psi_1}\triangleq&{\Big[{(1-\varphi_1)^2}\sqrt{\lambda_{\min}(p)}-\sqrt{c_1\lambda_{\max}(\xi)}\Big]}^{-1}\\
	&\times\sqrt{\lambda_{\max}(P)}{\varepsilon_N({{W}_m})}
	\end{split}
	\end{equation}
\end{theorem}
\begin{proof}
	First, let us define $\varphi_2$ and $\Theta_1$ as follows: 
	\begin{equation}\label{eq42n9}
	\begin{split}
	{\varphi_2}&\triangleq{{\varphi_1}}+\frac{{\sqrt{{c_1}\lambda_{\max}(\xi)}}}{(1-\varphi_1)\sqrt{\lambda_{\min}(P)}}\\
	{\Theta_1}&\triangleq{(1-\varphi_1)^{-1}{\sqrt{\lambda_{\max}(P)}}}.
	\end{split}
	\end{equation}
Moreover, we define $\mathcal{T}$ as the set of all time instants when the control packets are not lost within transmission. The set $\mathcal{T}$ is characterized as  
	\begin{equation}\label{eq17}
	\mathcal{T}\triangleq{\{t_n\}_{n\in{{\mathbb{N}}_0}}}\subseteq{{{\mathbb{N}}_0}},
	\end{equation}
	where ${t_{n+1}}>{t_n}, \forall{n}\in{{{{\mathbb{N}}_0}}}$. We denote the number of packet dropouts between $t_n$ and $t_{n+1}$ by $q_n$. It can be implied from the definition of $q_n$ that  
	\begin{equation}\label{eq18}
	{q_n}={t_{n+1}}-{t_n}-1,\qquad{\forall{n}\in{{{\mathbb{N}}_0}}}.
	\end{equation}
	 Assume that the current time step is $t_n$, ${n\in{{{\mathbb{N}}_0}}}$, and the control packet ${U}(x(t_n))$ arrives at the actuator. Between $t_n$ and  ${t_{n+1}}$, the actuator applies ${u_0}(x({t_n})),\dots,{u_{{q_n}}(x({t_n}))}$, computed according to the $\ell^2$-constrained ${\ell}^0$ sparse PPC policy (\ref{eq7}), to the plant successively. Then it follows from the plant dynamics (\ref{eq1}), the recursion used for the state predictions, (\ref{eq42n9}) and Lemma~\ref{lemma6} that $\sqrt{V_P(x(k))}$ is bounded as
	\begin{equation}\label{eq42n11}
	\sqrt{V_P(x(k))}\leq{{\varphi_2}\sqrt{{V_P(x({t_n}))}}+{\Theta_1}}{\varepsilon_N({{W}_m})}
	\end{equation}
	for every $k\in\{{t_n}+1,{t_n}+2,\dots,{t_n}+{q_n}+1\}$.  Note that (\ref{eq42}) holds for every realization $\bar{w}_0^{i}\in{\mathbb{W}^{i+1}}$, $\forall{i}\in\{0,1,\dots,N-1\}$. We can conclude from from (\ref{eq18}) and (\ref{eq42n11}) that 
	\begin{equation}\label{eq42n12}
	\sqrt{V_P(x(t_{n+1}))}\leq{{\varphi_2}\sqrt{{V_P(x({t_n}))}}+\Theta_1{\varepsilon_N({{W}_m})}}.
	\end{equation}
According to Assumption~\ref{ass1}, ${U}(x(0))$ is received by the actuator at time $0$. Given this and by mathematical induction, we can deduce from (\ref{eq42n12})	that the following holds:	
	\begin{equation}\label{eq42n13}
	\begin{split}
	\sqrt{V_P(x({t_{n}}))}\leq&{{{{\varphi}}_{2}^{n}}\sqrt{V_P(x(0))}}\\
	&{+(1+{\varphi_2}+\dots+{{\varphi}_{2}^{n-1}}){\Theta_1}{\varepsilon_N({{W}_m})}}.
	\end{split}
	\end{equation}
	It stems from (\ref{eq42n8}) ($0\leq\varphi_{2}<1$) and properties of weighted norms that
	\begin{equation}\label{eq42n14}
	\begin{split}
		\sqrt{V_P(x({t_{n}}))}\leq&{{{\varphi}_{2}^{n}}{\sqrt{\lambda_{\max}(P){\lVert{x(0)}\rVert}_{2}}}}\\
		&{+{(1-{{\varphi_2}})^{-1}}{\Theta_1}{\varepsilon_N({{W}_m})}}.
	\end{split}
	\end{equation}
	Now, we can use the inequality (\ref{eq42n11}) to derive
	\begin{equation}\label{eq42n15}
	\begin{split}
	\sqrt{V_P(x({k}))}\leq&{{{\varphi}_{2}^{n+1}}{\sqrt{\lambda_{\max}(P){\lVert{x(0)}\rVert}_{2}}}}\\
	&{+{(1-{{\varphi_2}})^{-1}}{\Theta_1}{\varepsilon_N({{W}_m})}}.
	\end{split}
	\end{equation}
	for any $k\in\{{t_n}+1,\dots,{t_{n+1}}\}$. Based on the properties of weighted norms, the following is extracted from (\ref{eq42n15}): 
	\begin{equation}\label{eq42n16}
	{{\lVert{x(k)}\rVert}_{2}}\leq{{{\varphi}_{2}^{{n}+1}}}\sqrt{\frac{{\lambda_{\max}(P){\lVert{x(0)}\rVert}_{2}}}{{{{\lambda}_{\min}}(P)}}}+\Psi_1.
	\end{equation}
	where $\Psi_1$ is specified by (\ref{eq42n8n1}). The inequality (\ref{eq42n16}) proves the bounded-ness of the plant state at each time instant $k\in\mathbb{N}_0$. Finally, using the fact $k\to\infty$ is equivalent to $n\to\infty$ and $0\leq\varphi_{2}<1$, we obtain (\ref{eq42n8n}) which completes the proof.     	     
\end{proof} 
\begin{remark}
Theorem~\ref{th2} characterizes the effect of the disturbance signal on the system performance. According to (\ref{eq41n5nn}) and (\ref{eq42n8n1}), $\Psi_1$ is an increasing function with respect to ${{W}_m}$. So, once one allows for larger disturbances by increasing ${{W}_m}$ (while keeping other parameters intact), the bound on the steady-state $\ell^2$ norm of the state grows larger. Since ${{W}_m}$ is a finite value, such increase does not affect system stability but it may lead to performance degradation.               
\end{remark}
\noindent
Although we will use an effective method for solving the optimization problem (\ref{eq7}) associated with $\ell^2$-constrained $\ell^0$ sparse PPC, this problem is still non-convex and computationally cumbersome. One remedy for that is mainly relaxing the $\ell^0$ norm of the control input by $\ell^1$ norm and setting quadratic penalties on the states. We formalize this as unconstrained ${\ell}^1$-${\ell}^2$ sparse PPC and study the stability of the NCS of Fig.~\ref{ppcfig1} under such control method in the following section.
\section{Unconstrained ${\ell}^1$-${\ell}^2$ Sparse PPC}
As already mentioned, the cost function based on which PPC works has the structure of (\ref{eq2}). For the specific case of unconstrained ${\ell}^1$-${\ell}^2$ sparse PPC, $T({\tilde{x}_{N}})={{\lVert{{\tilde{x}_{N}}}\rVert}_{P}^{2}}$ and $S({\tilde{x}_{i}},{u_i})=$ ${{\lVert{{\tilde{x}_{i}}}\rVert}_{Q}^{2}}+\nu{\lvert{u_{i}}\rvert}$  . We consider $\nu>0$, and $Q$ and $P$ as positive definite matrices. Remember that ${\tilde{x}_{i+1}}=A{\tilde{x}_{i}}+B{u_{i}}$, $i=0,...,N-1$, describes the recursion used for the future states prediction. Considering this and after some manipulations, we can reexpress the cost function $J$ as follows:
\begin{equation}\label{eq3}
J({x},{U})={{\lVert{M{U}-K{x}}\rVert}_{2}^{2}}+{{\lVert{{{x}}}\rVert}_{Q}^{2}}+\nu{{\lVert{{U}}\rVert}_{1}},
\end{equation} 
where matrices $M$ and $K$ are defined based on (\ref{eq4}) as in the previous section. Furthermore, $x={\tilde{x}_{0}}$ and ${U}=[{u_{0}}\dots {{u}_{(N-1)}}]^\top$. It follows from (\ref{eq3}) that every control packet $U({{x}}(k))$, $k\in{\mathbb{N}_0}$, is generated based on 
\begin{align}\label{eq5}
{U}(x)&=\arg\min_{U\in{\mathbb{R}}^{N}}{{\lVert{M{U}-Kx}\rVert}_{2}^{2}}+{{\lVert{{{x}}}\rVert}_{Q}^{2}}+\nu{{\lVert{{U}}\rVert}_{1}}.
\end{align}   
The solution to an optimization problem such as the one associated with (\ref{eq5}) is sparse\cite{hayashi2013user}. There are several approaches for solving such a problem. Under unconstrained ${\ell}^1$-${\ell}^2$ associated with (\ref{eq5}), practical stability can be obtained under specific conditions. We derive such conditions in terms of $Q>0$, $P>0$ and $\nu>0$ in what follows.    
\begin{remark}\label{rem2}
Based on the same line of arguments as in Remark~\ref{rem1n}, the  unconstrained ${\ell}^1$-${\ell}^2$ sparsity-promoting optimization considered here possesses the same formulation as its counterpart in \cite[Section~III-A]{nagahara2014sparse} in the unconstrained ${\ell}^1$-${\ell}^2$ sparse PPC of disturbance-free plants.  
\end{remark}
\noindent
For proving the stability analysis results, we utilize the value function 
\begin{equation}\label{eq8}
V(x)\triangleq{\min_{{U}\in{\mathbb{R}}^{N}}{J(x,U)}},
\end{equation}
where $J(x,U)$ represents the cost function in (\ref{eq3}). We start by  deriving bounds on $V(x)$ in the following lemma.  
\begin{lemma}\label{lemma1}
Consider ${\Pi}^{\star}$ and ${M^\dagger}$ as specified in (\ref{eq10}). Define the function $\tau(.)$ as  $\tau(y)\triangleq{{\alpha}y+({\beta}+{{\lambda}_{\max}}(Q)){y^2}}$ with ${\alpha}=\nu\sqrt{n}{{\sigma}_{\max}}({M}^{\dagger}K)$ and ${\beta}={{{\lambda}_{\max}}({\Pi}^{\star}})$. Then for every $x\in{{\mathbb{R}}^{n}}$, lower and upper bounds are derived on $V(x)$ as follows:
\begin{equation}\label{eq9}
{{{\lambda}_{\min}}(Q)}{{\lVert{x}\rVert}_{2}^{2}}\leq{V(x)}\leq{\tau({\lVert{x}\rVert}_{2}}).
\end{equation}
\end{lemma}
\begin{proof}
Considering Remark~\ref{rem2}, we can conclude the claim immediately from \cite[Lemma~5]{nagahara2014sparse}.
\end{proof}
\begin{lemma}\label{lemma2}
Let $P>{0}$ solve the Riccati equation (\ref{eq12}) with $r={{\mu}^{2}}N/{\zeta}$, $\zeta>{0}$. Define $\chi\triangleq{{\lambda_{\max}}(Q)+{\lambda_{\max}}({K^\top}K)}$. 
Then the following holds for every $x\in{\mathbb{R}}^{n}$ and every realization of $\bar{w}_0^{i-1}\in{\mathbb{W}^i}$:
\begin{equation}\label{eq13}
\begin{split}
\sqrt{V({f^i}(x,\bar{w}_0^{i-1}))}\leq&{\sqrt{V(x)-{{{\lambda_{\min}}(Q){\lVert{x}\rVert}_{2}^{2}}}}}\\
	&{+\sqrt{\chi}{\gamma_N({W_m})}+\sqrt{\zeta}},
\end{split}
\end{equation}
where ${i=1,2,\dots,N}$. 
\end{lemma}
\begin{proof}
It is implied from Remark~\ref{rem2} that the formulation of $J(x,U)$  in (\ref{eq3}) is the same across our paper and \cite[Section ~III-A]{nagahara2014sparse}.  Moreover, ${\bar{f}^i}(x)$ is equivalent to ${{f}^i}(x)$ in \cite[Section~IV-A]{nagahara2014sparse}. Having the latter statements in mind and considering $\sqrt{a+b}\leq{\sqrt{a}+\sqrt{b}}$, $\forall{a,b}\geq0$, we can deduce from \cite[Lemma~7]{nagahara2014sparse} that
\begin{equation}\label{eq13n3}
\sqrt{V({\bar{f}^i}(x))}\leq{{\sqrt{V(x)-{{\lambda}_{\min}}(Q){{\lVert{x}\rVert}_{2}^{2}}}+\sqrt{\zeta}}}
\end{equation}
holds for any $x\in{\mathbb{R}}^{n}$, ${i=1,2,\dots,N}$. It stems from the definition of $V(x)$ that $V(x)\leq{J(x,0)}$, $\forall{x\in\mathbb{R}^n}$. Then according to (\ref{eq3}) and properties of weighted norms, ${\sqrt{V(x)}\leq\sqrt{{\chi}}{{\lVert{x}\rVert}_{2}}}$ is deduced. Using (\ref{eq9}), we can derive
\begin{equation}\label{eq13n3n3}
\begin{split}
\sqrt{V({f^i}(x,\bar{w}_0^{i-1}))}-\sqrt{V({\bar{f}^i}(x))}&\leq{\sqrt{{\chi}}{{\lVert{{f^i}(x,\bar{w}_0^{i-1})}\rVert}_{2}}}\\
&{-{\sqrt{{\lambda_{\min}(Q)}}{{\lVert{{\bar{f}^i}(x)}\rVert}_{2}}}}
\end{split}	
\end{equation}
Since $\chi\geq{\lambda_{\min}(Q)}$ and according to the reverse triangle inequality, we reformulate the right-hand-side of (\ref{eq13n3n3}) as  
\begin{equation*}\label{eq13n3n4}
\sqrt{V({f^i}(x,\bar{w}_0^{i-1}))}-\sqrt{V({\bar{f}^i}(x))}\leq{\sqrt{{\chi}}{{\lVert{{f^i}(x,\bar{w}_0^{i-1})-{{\bar{f}^i}(x)}}\rVert}_{2}}}.
\end{equation*}
Now, based on Lemma~\ref{lemma0} and (\ref{eq11n2}), we can conclude that 
\begin{equation}\label{eq13n3n5}
\sqrt{V({f^i}(x,\bar{w}_0^{i-1}))}-\sqrt{V({\bar{f}^i}(x))}\leq{\sqrt{{\chi}}{\gamma_N({W_m})}}
\end{equation}
holds for every realization of $\bar{w}_0^{i-1}\in{\mathbb{W}^i}$ and every $x\in{\mathbb{R}}^{n}$. The claim follows immediately by adding up the corresponding sides of (\ref{eq13n3}) and (\ref{eq13n3n5}).	
\end{proof}
\noindent
\begin{lemma}\label{lemma3}
	Suppose that $P>0$ is a solution to the Riccati equation (\ref{eq12}) with $r={{\mu}^{2}}N/{\zeta}$ where $\zeta>{0}$ is an arbitrary positive real number. Then real constant $\varphi{\in(0,1)}$ exists in such a way that 
	\begin{equation}\label{eq14}
	\begin{split}
	\sqrt{V({f^i}(x,\bar{w}_0^{i-1}))}\leq&{{\sqrt{\varphi{V(x)}}+{\frac{\sqrt{\lambda_{\min}(Q)}}{2}}}}\\
	&+{\sqrt{\chi}{\gamma_N({W_m})}+\sqrt{\zeta}}	
	\end{split}
	\end{equation}
	holds for every $x\in{{\mathbb{R}}^n}$, every realization of $\bar{w}_0^{i-1}\in{\mathbb{W}^i}$ and every $i$ belonging to the set $\{1,2,\dots,N\}$.   
\end{lemma}

\begin{proof}
	Following the same steps as in the proof of \cite[Lemma~8]{nagahara2014sparse} leads us towards deriving the claim. This is due to the fact that we use the same formulation for the cost function and recursion for state prediction as the ones utilized in \cite[Section~III-A]{nagahara2014sparse} for the unconstrained ${\ell}^1$-${\ell}^2$ sparse PPC of disturbance-free plants.  
\end{proof}   
\noindent
Now, we are ready to establish the sufficient conditions for practical stability of the NCS of Fig.~\ref{ppcfig1} under the unconstrained ${\ell}^1$-${\ell}^2$ sparse PPC. The following theorem presents the latter stability conditions.
\begin{theorem}\label{th1}
Suppose $\zeta>{0}$ exists in such a way that $P>0$ satisfies (\ref{eq12}) with $r={{\mu}^{2}}N/{\zeta}$. Then for the NCS of Fig.~\ref{ppcfig1} controlled according to the unconstrained ${\ell}^1$-${\ell}^2$ sparse PPC (\ref{eq5}), the ${\ell}^2$ norm of $x(k)$ is bounded at each time instant $k\in\mathbb{N}_0$ and 
\begin{equation}\label{eq15}
\lim\limits_{k\to\infty}{{\lVert{x(k)}\rVert}_{2}}\leq{\Psi}
\end{equation}
holds where 
 	\begin{equation}\label{eq15s}
 	{\Psi\triangleq{{{\frac{1}{1-\sqrt\varphi}}}\Big[\frac{1}{2}}+\sqrt{\frac{\zeta}{{{{\lambda}_{\min}}(Q)}}}+{\frac{\sqrt{\chi}{\gamma_N({W_m})}}{{{\sqrt{{\lambda}_{\min}(Q)}}}}}\Big]}
 	\end{equation}
 	and 
 	\begin{equation}\label{eq16}
 \varphi\triangleq{1-{{{{\lambda}_{\min}}(Q)}}({\alpha}+{\beta}+{{{\lambda}_{\max}}(Q)})^{-1}}.
 \end{equation}	
 In (\ref{eq16}), $\alpha$ and $\beta$ are characterized as in Lemma~\ref{lemma1}.
\end{theorem}	   
\begin{proof}
Recall the definition and properties of $t_n$ and $q_n$ from the proof of Theorem~\ref{th2}. Consider $t_n$ as the current time step. Based on the recursion used for the prediction of the future states, dynamics of the plant (\ref{eq1}) and Lemma~\ref{lemma3}, we can derive the following upper bound on the square root of the value function of $x(k)$:
 \begin{equation}\label{eq19}
 \sqrt{V(x(k))}\leq{\sqrt{\varphi{V(x({t_n}))}}+\Theta}
 \end{equation}  
 	for every $k\in\{{t_n}+1,{t_n}+2,\dots,{t_n}+{q_n}+1\}$. In (\ref{eq19}), ${\Theta}$ is defined as ${\Theta}\triangleq{{\sqrt{\lambda_{\min}(Q)}}/{2}}+\sqrt{\chi}{\gamma_N({W_m})}+\sqrt{\zeta}$. Note that (\ref{eq14}) holds for every realization $\bar{w}_0^{i-1}\in{\mathbb{W}^i}$, $\forall{i}\in\{1,2,\dots,N\}$.	Based on ${t_{n+1}}={{t_n}+{q_n}+1}$ and (\ref{eq19}), we can derive 
 	\begin{equation}\label{eq20}
 	\sqrt{V(x(t_{n+1}))}\leq{\sqrt{\varphi{V(x({t_n}))}}+\Theta}.
 	\end{equation}
 	We know from Assumption~\ref{ass1} that ${U}(x(0))$ reaches the actuator at time $k=0$. Based on this and by mathematical induction, $V(x({t_n}))$ is bounded as 
 		\begin{equation*}\label{eq26}
 	\sqrt{V(x({t_{n}}))}\leq{{{\sqrt{\varphi}}^{n}}\sqrt{V(x(0))}+(1+{\sqrt\varphi}+\dots+{{\sqrt\varphi}^{n-1}})\Theta}.
 	\end{equation*}
 	Then according to Lemma~\ref{lemma1}, we have
 	\begin{equation}\label{eq27}
 	\sqrt{V(x({t_{n}}))}\leq{{{\sqrt\varphi}^{n}}{\sqrt{\tau({\lVert{x(0)}\rVert}_{2}})}+{(1-{{\sqrt\varphi}})^{-1}}\Theta}.
 	\end{equation}
 Now, it stems from (\ref{eq19}) that	
 	\begin{equation}\label{eq35}
 \sqrt{V(x({k}))}\leq{{{\sqrt\varphi}^{n+1}}{\sqrt{\tau({\lVert{x(0)}\rVert}_{2}})}+{(1-{{\sqrt\varphi}})^{-1}}\Theta}
 \end{equation}
 holds for any $k\in\{{t_n}+1,\dots,{t_{n+1}}\}$. Considering the lower bound on $V(x)$ in Lemma~\ref{lemma1}, we establish
 \begin{equation}\label{eq37}
 {{\lVert{x(k)}\rVert}_{2}}\leq\sqrt{{{\varphi}^{{n}+1}}}\sqrt{\frac{{\tau({\lVert{x(0)}\rVert}_{2}})}{{{{\lambda}_{\min}}(Q)}}}+\Psi,
 \end{equation}
 which implies the boundedness of the plant state at each time instant $k\in\mathbb{N}_0$. Moreover, $\Psi$ is defined as in (\ref{eq15s}). By letting $k$ go to infinity, which is equivalent to $n$ going to infinity, we derive (\ref{eq15}) and the proof will be complete.  
\end{proof}	
\begin{remark}
The bound derived in (\ref{eq15}) on the steady-state $\ell^2$ norm of the plant state characterizes the impact of disturbance signal on system performance. According to (\ref{eq15s}) and (\ref{eq13n4n2n}),   $\Psi$ is an increasing function of ${W_m}$. Thus, rendering $W_m$ larger while keeping other parameters untouched will lead to an increase in $\Psi$. Since $W_m$ is a finite number, this increase does not affect stability but my degrade systems performance.
\end{remark}

\section{Simulation Example}
Here, we simulate the NCS of Fig.~\ref{ppcfig1} via applying the sparse PPC policies developed in the previous sections. To do so, first we need to consider a plant model following (\ref{eq1}). This model is characterized by state matrix
\begin{equation}\label{eq46}
A=
\begin{bmatrix}
 0.3966&-0.4586&-0.0250&-0.7958\\
0.7459&0.8061&-0.0983&0.7943\\
-0.9451&-0.3111&-0.8236&0.2473\\
0.1551&-1.3821&-1.9151&0.0369
\end{bmatrix},
\end{equation}
the input vector
\begin{equation}\label{eq47}
B=
\begin{bmatrix}
1.0617\quad
-0.1986\quad
-0.3184\quad
0.5562 
\end{bmatrix}^\top
\end{equation}
and a disturbance signal with i.i.d samples taken from a uniform distribution over $[-W_m,W_m]$. 
\begin{figure}
	\centering
	\includegraphics[width=\columnwidth,height=7cm]{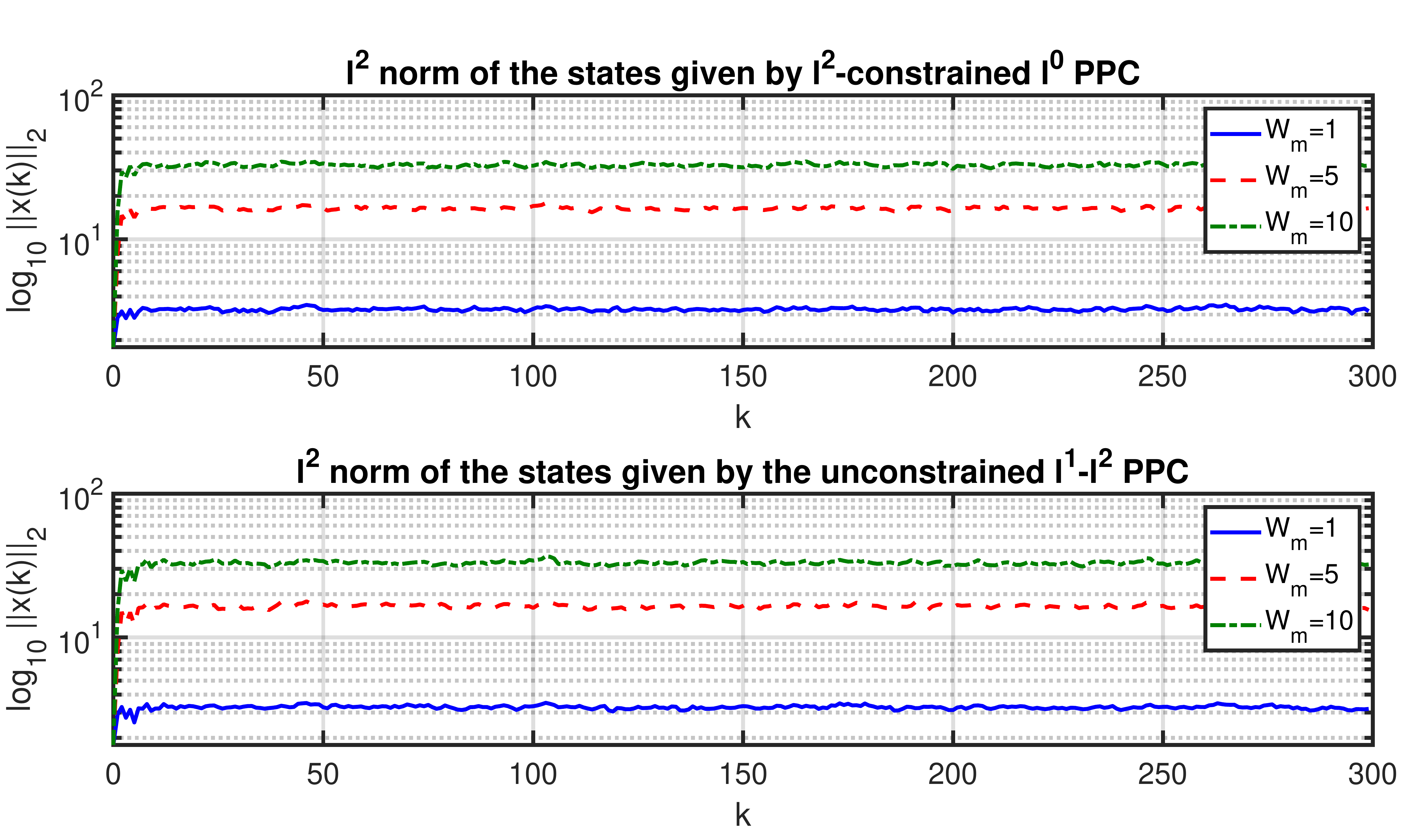}
	\caption{Average ${\ell^2}$ norm of the state $x(k)$ in the unconstrained ${\ell}^1$-${\ell}^2$ PPC (top) and ${\ell}^2$-constrained ${\ell}^0$ PPC (bottom)}
	\label{figppc2}
\end{figure}
The elements of $A$ and $B$ are samples of a normal distribution with mean $0$ and variance $1$. In (\ref{eq46}) and (\ref{eq47}), $A$ and $B$ form a pair whose reachability is straightforward to show. Note that the state matrix $A$ has $2$ unstable eigenvalues and $2$ eigenvalues inside the unit circle. The initial state vector $x(0)$ is comprised of i.i.d elements with normal distribution. We simulate the channel by setting the process $l$ in a way that the number of consecutive packet dropouts has uniform distribution with the support set $\{0,1,..,N-1\}$. For the controller simulation, we consider the prediction horizon $N$ as $N=10$ and the weighting matrix $Q$ as $Q=I_{4\times4}$ in both cases of ${\ell}^2$-constrained ${\ell}^0$ sparse PPC and unconstrained ${\ell}^1$-${\ell}^2$ sparse PPC. For simulating the ${\ell}^2$-constrained ${\ell}^0$ sparse PPC, we regulate the design parameters $\Pi$ and $\xi$ based Theorem~\ref{th2}. For this example, we set $\xi$ as $\xi=\big({(1-\varphi_1)^4{\lambda_{\min}(P)}}/{4{c_1}}\big)I_{4\times4}$. Moreover, we implement OMP as the algorithm for solving the optimization problem related to (\ref{eq7}). For simulating the ${\ell}^1$-${\ell}^2$ sparse PPC, the parameters $\nu$ and $r$  are set as  $\nu=200$ and $r=2$. In this case, fast iterative shrinkage-
thresholding algorithm (FISTA) is implemented for solving the corresponding optimization problem. To evaluate the effect of the plant disturbance on the system performance, we carry out the simulation for $3$ different values of 
$W_m$, i.e., $W_m\in\{1,5,10\}$.

The results of the simulation are illustrated by Fig.~\ref{figppc2} and Fig.~\ref{figppc3}. Such results are obtained by averaging over $200$ number of $300$-sample-long simulations. The vertical axis of Fig.~\ref{figppc2} corresponds to ${{\lVert{x(k)}\rVert}_{2}}$ and the horizontal axis to time. As observed from Fig.~\ref{figppc2}, the $\ell^2$ norm of the state varies over a bounded range along the simulation time in both cases of ${\ell}^2$-constrained ${\ell}^0$ sparse PPC and unconstrained ${\ell}^1$-${\ell}^2$ sparse PPC. This shows that the proposed sparse PPC designs can render the NCS of Fig.~\ref{ppcfig1} practically stable. According to Fig~\ref{figppc2}, such stability holds regardless of the magnitude of $W_m$. Moreover, we can observe from Fig.~\ref{figppc2} that the states take larger values when $W_m$ is larger. Such performance degradation by increasing $W_m$ agrees with Remark~\ref{rem1n} and Remark~\ref{rem2}. Figure~\ref{figppc3} demonstrates the $\ell^0$ norm of the control packets $U(x(k))$. As curves in Fig.~\ref{figppc3} signify, indeed, the control packets generated based on both ${\ell}^2$-constrained ${\ell}^0$ sparse PPC and unconstrained ${\ell}^1$-${\ell}^2$ sparse PPC are sparse. Of course by manipulating $\nu$ in (\ref{eq3}) and $\Pi$ in (\ref{eq7}), one can regulate the sparsity of the control actions in unconstrained ${\ell}^1$-${\ell}^2$ sparse PPC and ${\ell}^2$-constrained ${\ell}^0$ sparse PPC, respectively. However, the trade-off existing between the system performance and sparsity should be taken into account. Note that $A$ in (\ref{eq46}) presents an unstable dynamics. Another interesting observation is that in the case of unconstrained ${\ell}^1$-${\ell}^2$ sparse PPC, the control input becomes less sparse as the disturbance amplitude grows larger and $Q$ and $\nu$ are kept fixed.

\begin{figure}
	\centering
	\includegraphics[width=\columnwidth,height=7cm]{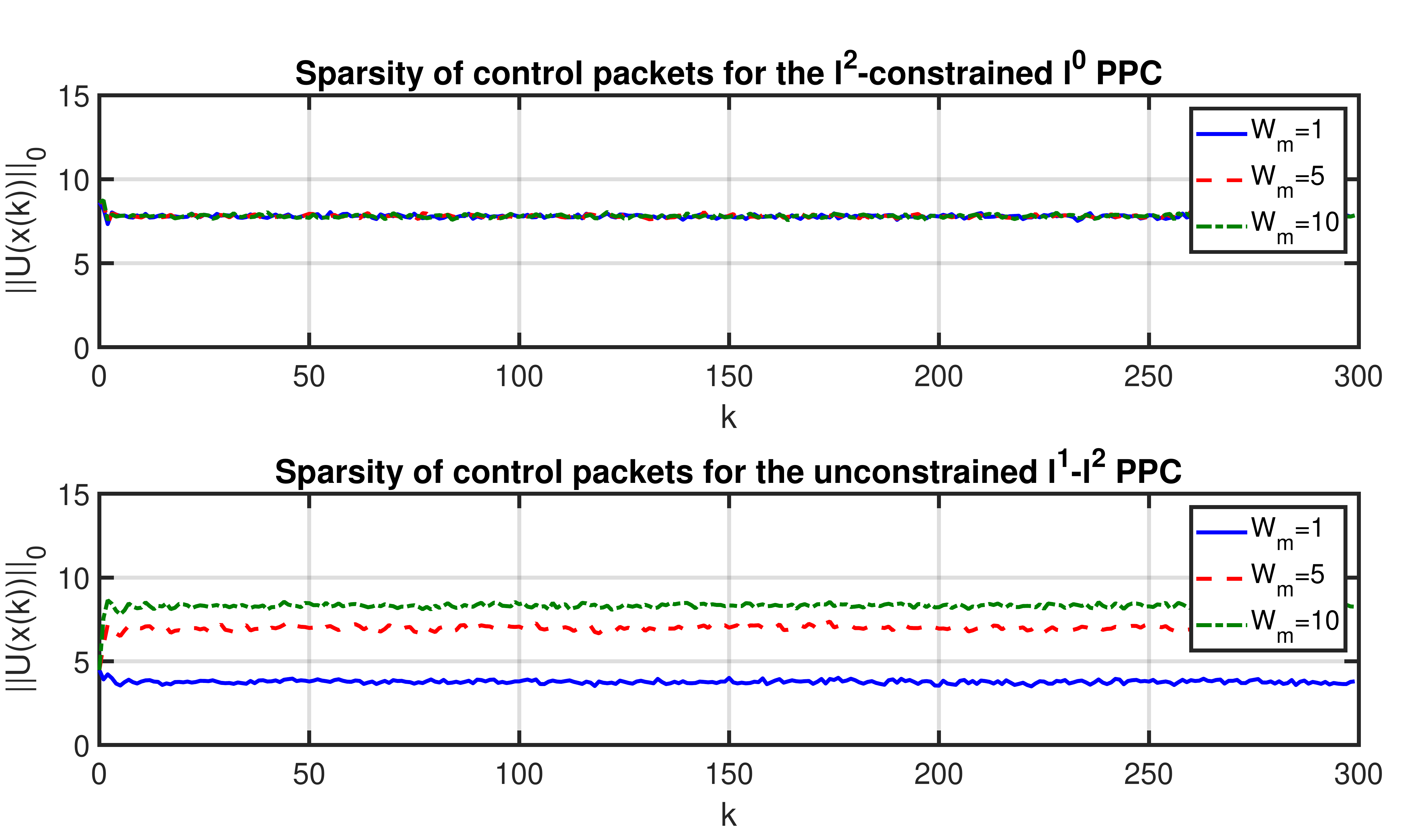}
	\caption{Average $\ell^0$ norm of the control packet ${U_h}(x(k))$ in the unconstrained ${\ell}^1$-${\ell}^2$ PPC (top) and ${\ell}^2$-constrained ${\ell}^0$ PPC (bottom)}
	\label{figppc3}
\end{figure}

\section{conclusion}
This paper has investigated the stability of linear discrete-time plants with bounded disturbances controlled based on sparse PPC. The control occurs in a feedback loop where the controller and the actuator communicate via a digital channel with data packet dropouts. The sparse PPC strategies taken into account perform based upon unconstrained ${\ell}^1$-${\ell}^2$ and ${\ell}^2$-constrained ${\ell}^0$ sparsity promoting optimizations. We have established the conditions of practical stability for both ${\ell}^1$-${\ell}^2$ sparse PPC and ${\ell}^2$-constrained ${\ell}^0$ sparse PPC. The bounds we have derived on the plant state helps interpreting the effect of disturbance on system stability. Such bounds are increasing functions of the disturbance amplitude in both cases of ${\ell}^1$-${\ell}^2$ sparse PPC and ${\ell}^2$-constrained ${\ell}^0$ sparse PPC. Through simulation, we have illustrated that with the proposed sparse PPC design, practical stability can indeed be obtained by sparse control packets. Moreover, in each case of ${\ell}^1$-${\ell}^2$ sparse PPC and ${\ell}^2$-constrained ${\ell}^0$ sparse PPC, enlarging the amplitude of the disturbance signal leads to performance degradation without jeopardizing practical stability.   
\appendices
\bibliographystyle{IEEEtran}
\bibliography{refs}

\begin{thebibliography}{10}
\providecommand{\url}[1]{#1}
\csname url@samestyle\endcsname
\providecommand{\newblock}{\relax}
\providecommand{\bibinfo}[2]{#2}
\providecommand{\BIBentrySTDinterwordspacing}{\spaceskip=0pt\relax}
\providecommand{\BIBentryALTinterwordstretchfactor}{4}
\providecommand{\BIBentryALTinterwordspacing}{\spaceskip=\fontdimen2\font plus
\BIBentryALTinterwordstretchfactor\fontdimen3\font minus
  \fontdimen4\font\relax}
\providecommand{\BIBforeignlanguage}[2]{{%
\expandafter\ifx\csname l@#1\endcsname\relax
\typeout{** WARNING: IEEEtran.bst: No hyphenation pattern has been}%
\typeout{** loaded for the language `#1'. Using the pattern for}%
\typeout{** the default language instead.}%
\else
\language=\csname l@#1\endcsname
\fi
#2}}
\providecommand{\BIBdecl}{\relax}
\BIBdecl

\bibitem{zhang2013network}
L.~Zhang, H.~Gao, and O.~Kaynak, ``Network-induced constraints in networked
  control systems--{A} survey,'' \emph{IEEE Transactions on Industrial
  Informatics}, vol.~9, no.~1, pp. 403--416, 2013.

\bibitem{zhang2019networked}
X.-M. Zhang, Q.-L. Han, X.~Ge, D.~Ding, L.~Ding, D.~Yue, and C.~Peng,
  ``Networked control systems: a survey of trends and techniques,''
  \emph{IEEE/CAA Journal of Automatica Sinica}, vol.~7, no.~1, pp. 1--17, 2019.

\bibitem{tang2006compensation}
P.~L. Tang and C.~W. de~Silva, ``Compensation for transmission delays in an
  ethernet-based control network using variable-horizon predictive control,''
  \emph{IEEE transactions on control systems technology}, vol.~14, no.~4, pp.
  707--718, 2006.

\bibitem{quevedo2007packetized}
D.~E. Quevedo, E.~I. Silva, and G.~C. Goodwin, ``Packetized predictive control
  over erasure channels,'' in \emph{2007 American Control Conference}.\hskip
  1em plus 0.5em minus 0.4em\relax IEEE, 2007, pp. 1003--1008.

\bibitem{quevedo2011packetized}
D.~E. Quevedo, J.~{\O}stergaard, and D.~Ne{\v s}i{\'c}, ``Packetized predictive
  control of stochastic systems over bit-rate limited channels with packet
  loss,'' \emph{IEEE Transactions on Automatic Control}, vol.~56, no.~12, pp.
  2854--2868, 2011.

\bibitem{quevedo2012robust}
D.~E. Quevedo and D.~Ne{\v{s}}i{\'c}, ``Robust stability of packetized
  predictive control of nonlinear systems with disturbances and markovian
  packet losses,'' \emph{Automatica}, vol.~48, no.~8, pp. 1803--1811, 2012.

\bibitem{peters2016shaped}
E.~G. Peters, D.~E. Quevedo, and J.~{\O}stergaard, ``Shaped {Gaussian}
  dictionaries for quantized networked control systems with correlated
  dropouts,'' \emph{IEEE Transactions on Signal Processing}, vol.~64, no.~1,
  pp. 203--213, 2016.

\bibitem{da2020hybrid}
F.~E. da~Silva, A.~L.~V. Iaremczuk, R.~D. Souza, G.~Brante, G.~L. Moritz, and
  S.~Hussain, ``Hybrid arq in wireless packetized predictive control,''
  \emph{IEEE Sensors Letters}, vol.~4, no.~5, pp. 1--4, 2020.

\bibitem{dunham1974automatic}
B.~Dunham, ``Automatic on/off switching gives 10-percent gas saving,''
  \emph{Popular Science}, vol. 205, no.~4, p. 170, 1974.

\bibitem{xiao2020modeling}
Z.~Xiao, Q.~Wang, P.~Sun, B.~You, and X.~Feng, ``Modeling and energy-optimal
  control for high-speed trains,'' \emph{IEEE Transactions on Transportation
  Electrification}, vol.~6, no.~2, pp. 797--807, 2020.

\bibitem{anderson2007optimal}
B.~D. Anderson and J.~B. Moore, \emph{Optimal control: {L}inear quadratic
  methods}.\hskip 1em plus 0.5em minus 0.4em\relax Courier Corporation, 2007.

\bibitem{nagahara2016maximum}
M.~Nagahara, D.~E. Quevedo, and D.~Ne{\v{s}}i{\'c}, ``Maximum hands-off
  control: {A} paradigm of control effort minimization,'' \emph{IEEE
  Transactions on Automatic Control}, vol.~61, no.~3, pp. 735--747, 2016.

\bibitem{nagahara2016discrete}
M.~Nagahara, J.~{\O}stergaard, and D.~E. Quevedo, ``Discrete-time hands-off
  control by sparse optimization,'' \emph{EURASIP Journal on Advances in Signal
  Processing}, vol. 2016, no.~1, p.~76, 2016.

\bibitem{chatterjee2016characterization}
D.~Chatterjee, M.~Nagahara, D.~E. Quevedo, and K.~M. Rao, ``Characterization of
  maximum hands-off control,'' \emph{Systems \& Control Letters}, vol.~94, pp.
  31--36, 2016.

\bibitem{ikeda2018maximum}
T.~Ikeda, M.~Nagahara, and K.~Kashima, ``Maximum hands-off distributed control
  for consensus of multiagent systems with sampled-data state observation,''
  \emph{IEEE Transactions on Control of Network Systems}, vol.~6, no.~2, pp.
  852--862, 2018.

\bibitem{kishida2018hands}
M.~Kishida, M.~Barforooshan, and M.~Nagahara, ``Hands-off control for
  discrete-time linear systems subject to polytopic uncertainties,''
  \emph{IFAC-PapersOnLine}, vol.~51, no.~23, pp. 355--360, 2018.

\bibitem{ikeda2020maximum}
T.~Ikeda and M.~Nagahara, ``Maximum hands-off control with time-space
  sparsity,'' \emph{IEEE Control Systems Letters}, vol.~5, no.~4, pp.
  1213--1218, 2020.

\bibitem{nagahara2011sparse}
M.~Nagahara and D.~E. Quevedo, ``Sparse representations for packetized
  predictive networked control,'' \emph{IFAC Proceedings Volumes}, vol.~44,
  no.~1, pp. 84--89, 2011.

\bibitem{nagahara2012packetized}
M.~Nagahara, D.~E. Quevedo, and J.~{\O}stergaard, ``Packetized predictive
  control for rate-limited networks via sparse representation,'' in
  \emph{Proceedings of the 51st IEEE Conference on Decision and Control (CDC)},
  Dec. 2012, pp. 1362--1367.

\bibitem{nagahara2014sparse}
------, ``Sparse packetized predictive control for networked control over
  erasure channels,'' \emph{IEEE Transactions on Automatic Control}, vol.~59,
  no.~7, pp. 1899--1905, 2014.

\bibitem{barforooshan2019sparse}
M.~Barforooshan, M.~Nagahara, and J.~{\O}stergaard, ``Sparse packetized
  predictive control over communication networks with packet dropouts and time
  delays,'' in \emph{2019 IEEE 58th Conference on Decision and Control
  (CDC)}.\hskip 1em plus 0.5em minus 0.4em\relax IEEE, 2019, pp. 8272--8277.

\bibitem{pati1993orthogonal}
Y.~C. Pati, R.~Rezaiifar, and P.~S. Krishnaprasad, ``Orthogonal matching
  pursuit: {R}ecursive function approximation with applications to wavelet
  decomposition,'' in \emph{Proceedings of the 27th Asilomar Conference on
  Signals, Systems and Computers}, vol.~1, Nov 1993, pp. 40--44.

\bibitem{bertsekas1976dynamic}
D.~P. Bertsekas and B.~DP, ``Dynamic programming and stochastic control.''
  1976.

\bibitem{poznyak2009advanced}
A.~S. Poznyak, \emph{Advanced mathematical tools for automatic control
  engineers: Stochastic techniques}.\hskip 1em plus 0.5em minus 0.4em\relax
  Elsevier, 2009.

\bibitem{hayashi2013user}
K.~Hayashi, M.~Nagahara, and T.~Tanaka, ``A user's guide to compressed sensing
  for communications systems,'' \emph{IEICE transactions on communications},
  vol.~96, no.~3, pp. 685--712, 2013.

\end{thebibliography}
\end{document}